\documentclass[10pt]{elsarticle}

\usepackage{amssymb}
\usepackage{amsmath}
\usepackage{amsthm}
\usepackage{graphicx}
\usepackage{graphics}
\usepackage{bm}

\newtheorem{lem}{Lemma}
\newtheorem{defn}{Definition}
\newtheorem{thm}{Theorem}

\journal{arXiv.org}

\begin{document}

\begin{frontmatter}

\title{Adaptive Network Dynamics and \\ Evolution of Leadership in Collective Migration}

%% use optional labels to link authors explicitly to addresses:
%% \author[label1,label2]{<author name>}
%% \address[label1]{<address>}
%% \address[label2]{<address>}

\author[mae]{Darren Pais}
\ead{dpais@alumni.princeton.edu}
 
\author[mae]{Naomi  E. Leonard\corref{cor1}}
\ead{naomi@princeton.edu}
\cortext[cor1]{Corresponding author.  Tel: 1-609-258-5129.  Fax:  1-609-258-6109.}

\address[mae]{Department of Mechanical and Aerospace Engineering, Princeton University, Princeton, NJ, 08544, U.S.A.}

\begin{abstract}
The evolution of leadership in migratory populations depends not only on costs and benefits of leadership investments but also on the  opportunities for individuals to rely on cues from others through social interactions.   We derive an analytically tractable adaptive dynamic network model of collective migration with  fast timescale migration dynamics and slow timescale adaptive dynamics of individual leadership investment and social interaction.   For large populations, our analysis of bifurcations with respect to investment cost explains the observed hysteretic effect associated with recovery of migration in fragmented environments.   Further, we show a minimum connectivity threshold above which there is evolutionary branching into leader and follower populations.   For small populations, we show how the topology of the underlying social interaction network influences the emergence and location of leaders in the adaptive system.   Our model  and analysis can describe other adaptive network dynamics involving collective tracking or collective learning of a noisy, unknown signal, and likewise can inform the design of robotic networks where agents use decentralized strategies that balance direct environmental measurements with agent interactions.
\end{abstract}

\end{frontmatter}

%% Start line numbering here if you want
% \linenumbers

%%%%%%%%%%%%%%%%%%%%%%%%%%%%%%%%%%% 

\section{Introduction \label{s:intro}}

A great variety of species including birds, fish, invertebrates and mammals engage in collective migration \cite{shaw,insects,herb,birdmig}. The migratory process is often an adaptive response to conditions such as competition for resources in a dynamic environment, seasonal variability, and selection of new habitats for breeding \cite{shaw,guttal,torney,Fagan2012}. On the one hand, animals perform  migratory tasks by leveraging environmental cues such as nutrient and thermal gradients, magnetic fields, odor cues, or visual markers \cite{nav,odor,context}. Measuring these stochastic environmental signals is complicated and requires the investment of time and energy, as well as the development of necessary physiological and sensory machinery such as vision in insects and vertebrates \cite{insects} and chemical signaling in bacteria \cite{cellmig}. On the other hand, animals can perform migrating tasks by leveraging social cues from neighbors (nearby conspecifics) \cite{guttal,torney,Fagan2012}. By imitating invested neighbors (or neighbors of invested neighbors, etc.) using consensus processes such as cohesion and alignment, some animals in a group can migrate well without paying the costs associated with directly measuring and processing the environmental signal. 

The interplay between costly information acquisition from the environment and relatively less expensive social interactions with the group raises two important questions regarding leadership and social interactions in migratory populations.  Here, leadership means the influence of individuals who are informed about the environmental signal, e.g., by investing time and energy in taking a measurement.   Although these informed individuals are referred to as ``leaders'' and the remaining uninformed individuals as ``followers'', it is {\em not} assumed that leaders can be identified or can signal their information to others.

The first important question relates to the migratory performance of large groups in the presence of a limited number of leaders, i.e., can a subset of informed individuals effectively lead a large group? Couzin et al.~\cite{couzin_leader} address this question using individual-based simulations (involving social forces of attraction, repulsion and alignment among individuals) and demonstrate that in a group of socially interacting individuals, a small fraction of informed leaders can effectively determine the direction of travel of a large group of uninformed followers.   

The second important question relates to the evolution of leadership in collective migration, i.e., under what conditions is the coexistence of invested leaders and social followers stable in an evolutionarily sense? This question is especially relevant when the cost of investing in signal acquisition is sufficiently high; followers can leverage the investments made by leaders using social interactions without having to pay the investment costs themselves, but not all individuals can be followers if the group is to migrate successfully.  Guttal and Couzin~\cite{guttal}  address this question (also see related commentary \cite{migcom}) using evolutionary simulations and an individual-based model similar to that used in \cite{couzin_leader}; they show that the specialization of groups into coexisting leaders and followers (also known as branching) is a stable evolutionary outcome. 

Motivated by these questions and results, we develop an analytically tractable model of collective migration with which we can rigorously study the adaptive network dynamics associated with the evolution of collective migration and the emergence of leadership when leadership is costly and social interaction is relatively cheap.  We investigate the influence on group-level outcomes of the (evolving) network topology,  i.e., who is sensing and responding to whom within the group,  as a function of the cost of investing in the environmental signal.   Our model can be generalized to a broader set of adaptive network dynamics associated with a collective task, such as collective tracking or collective learning of a noisy, unknown signal \cite{poulakakisTAC,olshevskySICON}, carried out by agents with decentralized strategies that balance direct environmental measurements with social interactions.  

Torney et al.~\cite{torney} derive a mean-field approximation to the evolutionary model studied in \cite{guttal}, and using tools from evolutionary adaptive dynamics \cite{adyn1,adyn2} prove conditions for the branching of a migrating population into leader and follower groups.  The mean-field approach effectively prescribes an all-to-all social interaction topology between the individuals in order to reduce dimensionality, critical to the analysis in \cite{torney}.   However, the approach ignores the potentially important role of limited social interactions; indeed, it has been shown that network topology plays a critical role in determining outcomes in collective dynamics \cite{poulakakisTAC,olshevskySICON,paleyosc,george1,star1,lead1}.   Our model derives directly from the model of \cite{torney} with a key generalization to the case of limited interaction networks and a modification that allows individual fitnesses to be computed from a linear matrix equation as a function of the network topology  encoded by a  directed graph.

In our model each agent $i$ has a scalar strategy $k_i(t) \in [0,1]$ at time $t$ which defines how much it invests in the environmental signal; $k_i = 0$ means no investment and $k_i=1$ means full investment.   The strategy also determines how much it attends to social interaction: a higher $k_i$ implies a lower attention to measurements of neighbors, equivalently, the associated edges of the network graph are scaled by $(1-k_i)^2$.   The network dynamics have two timescales.  In the fast timescale the strategies $k_i$ are fixed, such that the stochastic migration dynamics and the fitnesses depend on a fixed interconnection topology.  In the slow timescale the strategies $k_i$ change according to evolutionary (or adaptive) dynamics and with them the investments and graph edge weights.

We present three main results that leverage the extension of the migration model of \cite{torney} to directed, limited social interaction topologies and the corresponding matrix equation for fitnesses (which replaces extensive Monte-Carlo simulations as used in \cite{torney,guttal}). 
Our first main result is a complete bifurcation analysis of the two-timescale dynamics as a function of investment cost in the case of a large population with an underlying network topology that is all-to-all;  our results explain previous observations that are initial condition dependent and demonstrate the hysteretic effect associated with losing and then recovering migration ability as described in \cite{guttal,guttal2}.   
Our second main result addresses the two-timescale dynamics in the case of a large population with an underlying network topology that is limited;  we find  a relatively small threshold in connectivity above which there is evolutionary branching and emergence of leaders.   

Our third main result addresses the case of a small population in which the slow  evolutionary dynamics of strategies $k_i$, based on replication and mutation, are replaced with individual greedy adaptive dynamics.  We show the critical role that the structure of the underlying network topology plays in determining the location of leaders in the adaptive network and in influencing bifurcations in the dynamics as a function of increasing cost.    This analysis is motivated in part by an interest in leveraging the mechanisms of evolved natural collectives in the design of decentralized protocols for collective motion and decision-making in robotic groups \cite{paleyosc,vincent}.

The paper is outlined as follows\footnote{See \cite{PaisThesis} for related text and figures.}. In Section \ref{s:model} we present our evolutionary migration model. We derive analytical results for fitnesses on the fast timescale in Section \ref{s:fast}.   We study the slow timescale  dynamics in the all-to-all limit in Section \ref{s:all} and for limited interconnections in Section \ref{s:limited}.  We focus on adaptive dynamic nodes in small networks in Section \ref{s:adaptive} and conclude in Section \ref{s:final}.

%%%%%%%%%%%%%%%%%%%%%%%%%%%%%%%%%%% 

\section{Model description \label{s:model}}

Our model is derived from the mean-field migration model in \cite{torney} with two key modifications; we explicitly account for a limited social interaction graph topology in the dynamics and we introduce a slightly modified social noise model to allow for analytical fitness computations as a function of graph topology and individual investments.  In the remainder of the paper we will refer to individuals, agents and nodes interchangeably, and likewise population and network.

Consider a set of $N$ agents indexed by $i\in\left\{1,\cdots,N\right\}$.  Let $x_i (t) \in\mathbb{R}$ be the direction of migration of agent $i$ at time $t$, and let $\mu\in\mathbb{R}$ be the ``true'' desirable direction of migration.
Accurate tracking of the  direction $\mu$ over time may correspond to benefits such as improvement in environmental conditions for foraging, predator evasion, early access to breeding grounds, etc. Following \cite{torney}, the stochastic dynamics of each agent are given by
\begin{equation}
dx_i = k_{i} dx_{Di} + (1-k_{i})dx_{Si},
\label{full}
\end{equation}
where $dx_{Di}$ and $dx_{Si}$ are the driven tracking and social consensus stochastic processes, respectively. The adaptive strategy $k_i\in[0,1]$  tunes the level of investment made by agent $i$ in the driven and social processes. When $k_i=1$,  agent $i$ is  fully invested in the tracking process and ignores social cues, while when $k_i=0$  agent $i$ exclusively leverages social cues without tracking the environmental signal. 
 
The driven process $dx_{Di}$ is modeled as an Ornstein-Uhlenbeck stochastic process \cite{ou,gardiner} of the form 
\begin{equation}
dx_{Di} = -k_{Di}(x_i - \mu)dt + \sigma_D dW_{Di}. 
\label{driven}
\end{equation}
The parameter $k_{Di}\geq0$ corresponds to the gain associated with tracking, $\sigma_D^2>0$ is the noise intensity associated with measuring the environmental signal $\mu$, and $dW_{Di}$ represents the standard Wiener increment. For $k_{Di}>0$, the process \eqref{driven} has a steady-state mean and variance given by
\begin{equation}
\lim_{t\rightarrow \infty} E\left[x_i\right]=\mu,\;\;\lim_{t\rightarrow \infty} E\left[(x_i-\mu)^2\right]=\frac{\sigma_D^2}{2k_{Di}}. 
\label{var}
\end{equation}
Higher values of tracking gain $k_{Di}$ result in lower steady-state variance in migration direction $x_i$, which corresponds to improved tracking.

The social consensus process $dx_{Si}$ is modeled using basic tools from graph theory \cite{con1,con2,con3}. Individuals are modeled as nodes on a directed social interconnection graph with underlying structure defined by adjacency matrix $A=[a_{ij}]\in\mathbb{R}^{N\times N}$. A directed edge in the graph from individual $i$ to individual $j$ is read as ``$i$ can sense $j$''. Let $\mathcal{N}_i$ denote the set of neighbors of individual $i$ (i.e., the set of agents that individual $i$ can sense), and let $\|\mathcal{N}_i\|$ denote the cardinality of this set (number of neighbors that individual $i$ can sense).   The adjacency matrix $A$ is given by
\begin{equation}
a_{ij}=\begin{cases}  \|\mathcal{N}_i\|^+ &\text{ if }j\in \mathcal{N}_i \\   0 &\text{ otherwise},  \end{cases}
\label{adj}
\end{equation}
where $ \|\mathcal{N}_i\|^+$ is the pseudoinverse of $\|\mathcal{N}_i\|$ ($\|\mathcal{N}_i\|^+=0$ when $\|\mathcal{N}_i\|= 0$, $\|\mathcal{N}_i\|^+=1/\|\mathcal{N}_i\|$ otherwise).
The Laplacian matrix of the graph corresponding to $A$ is given by $L=\text{diag}(A\bm{1})-A$, where $\bm{1}$ is a vector of ones of appropriate dimension.  

The social consensus process  $dx_{Si}$ depends on the underlying social interaction graph Laplacian $L$, the gain associated with the social process $k_{Si}\geq 0$, and the noise associated with measuring the social signal $\sigma_{Si}>0$ as follows:
\begin{equation}
\begin{aligned}
dx_{Si}&=-k_{Si}L_i \bm{x} dt + \sigma_{Si} dW_{Si},
\end{aligned}
\label{social}
\end{equation}
where $dW_{Si}$ is the standard Wiener increment, $L_i$ denotes the $i^{\text{th}}$ row of $L$ and $\bm{x}$ is the vector of the $x_i$.   In the social consensus process the graph that represents the social interactions is the underlying graph modified such that all edges from agent $i$ are weighted by gain $k_{Si}$, $i=1, \ldots, N$.  Then $k_{Si}L_i\bm{x}=k_{Si}\|\mathcal{N}_i\|^+\sum_{j\in\mathcal{N}_i}(x_i-x_j)$.   

Following the setup in \cite{torney}, we make a simplification to reduce the parameter space to one dimension by assuming that the gains are proportional to the relative investments in each process, i.e., $k_{Di} =k_i$ and $k_{Si} = 1-k_i$. Substituting into \eqref{driven} and \eqref{social} and then into \eqref{full} we have
\begin{equation}
\begin{aligned}
dx_i&=k_idx_{Di}+(1-k_i)dx_{Si}\\
&=-k_i^2(x_i-\mu)dt-(1-k_i)^2L_i\bm{x}dt+\sqrt{k_i^2\sigma_D^2+(1-k_i)^2\sigma_{Si}^2}\;dW_i,
\end{aligned}
\label{sys1}
\end{equation} 
where $dW_i$ is the standard Wiener increment. Define the coordinate transformation  $\tilde{x}_i$ as 
\begin{equation}
\tilde{x}_i =\frac{x_i-\mu}{\sigma_D}, \text{ and correspondingly }\bm{\tilde{x}}=\frac{\bm{x}-\mu\bm{1}}{\sigma_D}.
\label{trans}
\end{equation}
Substituting \eqref{trans} in \eqref{sys1} and using $L_i \bm{1}=0$ we have the normalized dynamics
\begin{equation}
d\tilde{x}_i=-k_i^2 \tilde{x}_i dt-(1-k_i)^2L_i\bm{\tilde{x}}dt+\sqrt{k_i^2 + (1-k_i)^2\frac{\sigma_{Si}^2}{\sigma_D^2}}\;dW_i.
\label{sys}
\end{equation} 

The social noise term $\sigma_{Si}$ reflects the difficulty that agents have in extracting social cues from interactions with neighbors. In \cite{torney}, it is assumed that this difficulty (magnitude of $\sigma_{Si}$) decreases as the ordering or coherence of the population increases. Here we take a slightly different local (and graph dependent) approach and relate the social noise term for an  agent to the average investment of the neighbors of that agent as
\begin{equation}
\frac{\sigma_{Si}^2}{\sigma_D^2} = \beta^2 (1 -k_{nbhd,i}),
\label{beta}
\end{equation}
where
\[
k_{nbhd,i} = \|\mathcal{N}_i\|^+ \sum_{j\in \mathcal{N}_i} k_i 
\]
and $\beta^2$ is a social noise scaling parameter. In vector form, $\bm{k}_{nbhd} =A\bm{k}$, $\bm{k} = (k_1, \ldots, k_N)^T$. 
In this model agents that interact socially with neighbors having a high level of investment, have a correspondingly lower social noise term, and are hence better able to extract social cues from their neighbors.   

The stochastic system \eqref{sys} can be written compactly in matrix form as 
\begin{equation}
d\tilde{\bm{x}} = -(K_1+K_2L)\tilde{\bm{x}} \;dt + S d\bm{W},
\label{sysm}
\end{equation}
where the diagonal matrices $K_1$, $K_2$ and $S$ are given by  $K_1 = \text{diag}(k_i^2)$, $K_2 = \text{diag}\left((1-k_i)^2\right)$ and $S = \text{diag}\left( \sqrt{k_i^2 + \beta^2(1-k_i)^2(1-k_{nbhd,i}) }\right)$. 

As discussed in \cite{guttal, torney}, the long-term migratory performance of  agent $i$ with dynamics \eqref{sysm} can be computed as
$\exp \left( \frac{-\sigma_{ss,i}^2}{2}\right)$, where $\displaystyle{\sigma_{ss,i}^2=\lim_{t\rightarrow \infty} E\left[(x_i-\mu)^2\right]}$ is the steady-state variance of $x_i$.   For an agent $i$ that moves at a constant mean speed in the direction $x_i$,  this performance measure corresponds to the expected migration speed of the agent in the desired direction $\mu$. The fitness or utility of  agent $i$  
is defined in \cite{guttal,torney} as
\begin{equation}
F_i=\exp \left( \frac{-\sigma_{ss,i}^2}{2}\right)\exp\left(-ck_{i}^2\right),
\label{fitsol}
\end{equation}
where the second exponential models the cost associated with investment in tracking, and $c>0$ is  a scaling cost parameter.  The  form of cost function is chosen for analytical tractability. Simulations in \cite{torney,guttal} show that reasonable variations of the fitness function \eqref{fitsol} yield qualitatively comparable results. 

The saturating form of the performance function $ \exp \left( \frac{-\sigma_{ss,i}^2}{2}\right)$ as a function of investment can be interpreted as modeling the diminishing returns of increasing investment. Further, the quadratic form of the cost $ck_{i}^2$ implies that higher investments in the driven process are increasingly costly. The optimal strategy for solitary migrating individuals can be found by  maximizing \eqref{fitsol} with respect to $k_{Di}$ after substituting for $\sigma^2_{ss,i}=\frac{\sigma_D^2}{2 k_{Di}}$ from \eqref{var}. This results in 
\begin{equation}
k_{D,opt}=\sqrt[3]{\frac{\sigma_D^2}{8c}}.
\label{optsol}
\end{equation}
While disconnected individuals may adopt the optimal strategy \eqref{optsol}, the presence of social interactions between individuals leads to very different strategies.

%%%%%%%%%%%%%%%%%%%%%%%%%%%%%%%%%%% 

\section{Fast timescale analysis \label{s:fast}}

In this section we study the stochastic migration dynamics \eqref{sysm} and derive analytical tools to compute migratory performance as a function of the underlying network graph Laplacian $L$ and agent investments $k_i$. These tools are used to compute fast timescale fitness (utility) in slow evolutionary (adaptive) dynamics in the sections that follow.  This includes Section \ref{s:all}, where the population is large and the underlying network topology is all-to-all, but most particularly  Section \ref{s:limited}, where the population is large and the underlying network topology is limited, and  Section \ref{s:adaptive} where the population is small.    

   It is straightforward to show that if $k_i>0$ for all $i$ then the population can migrate for any network topology; likewise, if the population can migrate then it is necessary that $k_i>0$ for at least one agent $i$ \cite{PaisThesis}.  If the network has a spanning tree and $k_i>0$  for the root node of the spanning tree then the population can migrate \cite{Moore2005}.  In Theorem \ref{th:necsuf} we prove conditions on  $L$ and $\bm{k}$ that are necessary and sufficient for asymptotic stability of the zero equilibrium of the dynamics \eqref{sysm} without noise:  
 \begin{equation}
 \dot{\tilde{\bm{x}}} = M\tilde{\bm{x}}, \text{ where }M=-(K_1+K_2 L) .
\label{nf}
 \end{equation}
The asymptotic stability of  \eqref{nf} corresponds to the population developing  the ability to collectively migrate since $\bm{\tilde x}\rightarrow \bm{0}\implies \bm{x} \rightarrow \mu \bm{1}$.  In Theorem~\ref{th:lyap}  we prove that asymptotic stability of \eqref{nf}  is necessary and sufficient  for there to exist a steady-state probability distribution of $\tilde{\bm{x}}$ for the stochastic dynamics \eqref{sysm}, and we show that the steady-state covariance matrix, and thus the individual fitnesses, can be computed from a Lyapunov equation that depends on $L$ and $\bm{k}$. 
The  dynamics  \eqref{nf} are asymptotically stable if and only if  $M$ is Hurwitz, i.e., all eigenvalues of $M$ have strictly negative real part.  

 The proof of Theorem \ref{th:necsuf} requires the following lemma from \cite{ren} (see also \cite{con1,fran}).  We first make the following definition.
\begin{defn}
\label{defspan}
A directed graph has a {\em directed spanning tree} 
if there exists at least one node $k$ on the graph such that a directed path exists from every other node on the graph to node $k$. Node $k$ is known as a root node of the graph. 
\end{defn}
\begin{lem}
\label{lm:ren}
For a general Laplacian matrix $\tilde{L}=[\tilde{l}_{ij}]\in\mathbb{R}^{N\times N}$ given by $\tilde{l}_{ij}\leq 0 $ for $i\neq j$ and $\displaystyle{\sum_{j=1}^N \tilde l_{ij}=0}$ for all $i\in\{1,\cdots,N\}$, the following conditions are equivalent:
\begin{enumerate}
\renewcommand{\labelenumi}{(\roman{enumi})}
\item $\tilde{L}$ has a simple zero eigenvalue and all of the other eigenvalues have positive real parts. 
\item The directed graph $\mathcal{G}(\tilde L)$ has a directed spanning tree, where $\mathcal{G}(\tilde L)$ is the graph with adjacency matrix $\tilde{A}=[\tilde{a}_{ij}]$, $\tilde{a}_{ii}=0$ for all $i$, and $\tilde{a}_{ij}=-\tilde{l}_{ij}$ for all $i\neq j$. 
\item For $\bm{z}\in\mathbb{R}^N$, the dynamics $\dot{\bm{z}}=-\tilde L\bm{z}$ converge asymptotically to $\alpha \bm{1}$ for some scalar $\alpha$. 
\end{enumerate}  
\end{lem}
\begin{proof}
See Lemma 3.1 in \cite{ren}, Theorem 2 in \cite{con1}, and Lemma 2 in \cite{fran}. 
\end{proof}

For undirected connected graphs, every node is a root node. For general (connected or disconnected) directed graphs, one can define a {\it root set} that is accessible from every other node in the network, i.e., there is a directed path from every node to at least one node in the root set. Let $\mathcal{R}(\tilde L)$ denote a minimal root set (set with smallest cardinality) of the graph $\mathcal{G}(\tilde L)$ with Laplacian $\tilde L$. The set $\mathcal{R}(\tilde L)$ is not necessarily unique as is illustrated in Figure~\ref{f:root}. For example, for an undirected connected graph, $\mathcal{R}(\tilde L)=\{i\}$ for any node $i$. 

\begin{figure}[ht!]
  \begin{center}
    \includegraphics[width=1\linewidth]{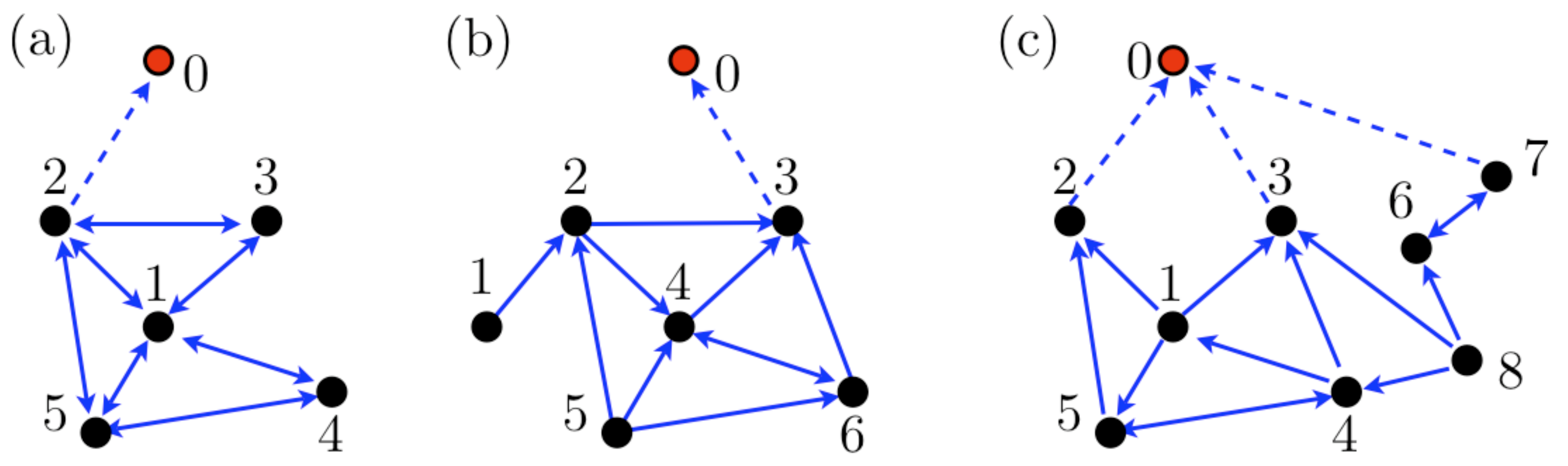}
    \caption{Illustrations of the root set $\mathcal{R}(L)$ and the conditions of Theorem \ref{th:necsuf}. In each graph, the set of nodes labeled $1,\cdots,N$ and solid arrows correspond to the underlying social  graph $\mathcal{G}(L)$. The complete set of nodes labeled $0,1,\cdots,N$ and all the arrows correspond to the augmented graph $\mathcal G(\hat{L})$, where node $0$ represents the external signal. All three augmented graphs shown have a spanning tree rooted at node $0$, and hence satisfy the conditions of Theorem \ref{th:necsuf}. (a) $\mathcal{R}(L)=\left\{i\right\},$ any $i\in \{1,2,3,4, 5\}$. (b) $\mathcal{R}(L)=\left\{3 \right\}$. (c)  $\mathcal{R}(L)=\left\{2,3,6\right\}\text{ or }\left\{2,3,7 \right\}$.  } 
    \label{f:root}
  \end{center}
\end{figure}

\begin{thm}
\label{th:necsuf}
Matrix $M$ from \eqref{nf} is Hurwitz if and only if there exists a minimal root set $\mathcal{R}(L)$ such that $k_j>0$ for all nodes $j\in\mathcal{R}(L)$, where $L$ is the Laplacian matrix of the underlying social graph with adjacency matrix \eqref{adj}.
\end{thm}
\begin{proof}
Define the normalized external signal state  $\tilde x_0=0$ (equivalent to $x_0=\mu$) and the augmented state vector $\bm{z}=\left[\tilde x_0\;\;\tilde{\bm x} \right]^T$. Consider the dynamics
\begin{equation}
\dot{\bm z} = \left[\begin{array}{c}\dot{\tilde {x}}_0 \\ \dot {\tilde{\bm x}}\end{array}\right]=-\left[\begin{array}{c|c}0 & 0_{N\times 1} \\\hline -K_1 \bm{1} & -M\end{array}\right]\left[\begin{array}{c}{\tilde {x}}_0 \\  {\tilde{\bm x}}\end{array}\right]=-\hat{L} \bm{z}.
\label{augd}
\end{equation}
Then $\hat L$ satisfies the properties of general Laplacian matrices given in Lemma \ref{lm:ren}.  $\mathcal G(\hat{L})$ is the same as the graph $\mathcal G(K_2 L)$ with the addition of the node  $0$ (with state ${\tilde x}_0$) having incoming links with weights $k_j^2$ from all nodes $j=1,\cdots,N$ (see Figure \ref{f:root} for an illustration where links to node $0$ with $k_j > 0$ are shown as dashed arrows). 

Since node $0$ of $\mathcal G(\hat L)$ has no outgoing links, following Definition \ref{defspan} we have the following condition,
\begin{equation}
\mathcal{G}(\hat L) \text{ has a spanning tree } \iff \mathcal R(\hat{L})=\left\{ 0 \right\}.
\label{root0}
\end{equation}
We claim the following  
\begin{equation}
\mathcal{G}(\hat L) \text{ has a spanning tree } \implies \exists\; \mathcal{R}(L) \text{ s.t. } k_j>0\text{ for all }j\in\mathcal{R}(L).
\label{if}
\end{equation}
We prove the statement above by contradiction. Assume that $\mathcal{G}(\hat L)$ has a spanning tree and for each root set $\mathcal{R}(L)$, there exists a node $j$ such that $k_j=0$. Since $\mathcal{G}(\hat L)$ has a spanning tree, $\mathcal{R}(\hat L)=\{0\}$, which means that there is a directed path from every node to node $0$. Now consider any root set $\mathcal{R}(L)$. Since $k_j=0$, node $j$ in $\mathcal{R}(L)$ can only reach node $0$ by a path to a node $m\notin \mathcal{R}(L)$, for which $k_m>0$. However, if such a path exists, then the set $\mathcal{R}(L)$ where node $j$ is replaced with node $m$ is another root set. By assumption we must have $k_m=0$. Thus there exists no directed path from node $j$ to node $0$. Hence $\mathcal{G}(\hat{L})$ does not have a spanning tree and we have proved the claim.

Consider any root set $\mathcal R (L)$ and assume that $k_j>0$ for all $j\in\mathcal R (L)$. Then all nodes $j\in\mathcal R (L)$ are connected to node $0$ of $\mathcal G(\hat L)$. For all nodes $m \notin \mathcal R(L)$, either $k_m>0$ and $m$ has a direct link to node $0$, or $k_m=0$ in which case $m$ has a link to at least one other node on a directed path to the root node $0$, via an element of $\mathcal R (L)$. Hence
\begin{equation}
\exists \; \mathcal R(L) \text{ s.t. }k_j>0\text{ for all }j\in\mathcal{R}(L)\implies \mathcal{G}(\hat L) \text{ has a spanning tree }  .
\label{onlyif}
\end{equation} 
Combining \eqref{root0}, \eqref{if} and \eqref{onlyif}, we have that 
\begin{align}
\exists \; \mathcal R(L) \text{ s.t. }k_j>0\text{ for all } j\in\mathcal{R}(L) &\iff \mathcal{G}(\hat L) \text{ has a spanning tree } \nonumber \\ 
& \iff \mathcal{R}(\hat{L})=\left\{0\right\} .
\label{co1}
\end{align} 

From Lemma \ref{lm:ren}, $\mathcal{G}(\hat L)$ has a spanning tree if and only if the dynamics \eqref{augd} converge asymptotically to $\alpha \bm{1}$ for some scalar $\alpha$. However in \eqref{augd}, the state $\tilde{x}_0=0$ is invariant and hence $\alpha=0$. 
Thus,
\begin{equation}
\mathcal{G}(\hat L) \text{ has a spanning tree } \iff M \text{ is Hurwitz}.
\label{co2}
\end{equation}
Combining \eqref{co1} and \eqref{co2} we have the desired result. 
\end{proof}
Each of the graphs illustrated in Figure~\ref{f:root} satisfies the conditions of Theorem~\ref{th:necsuf} since $k_j>0$ for all $j$ in at least one root set of $\mathcal R(L)$ (dashed arrows).   

We now return to the noisy migration model given by the system of stochastic equations \eqref{sysm}. In order to compute the fitness of an individual in a migratory collective as defined by \eqref{fitsol}, a computation of the steady-state variance of the individual's dynamics $\sigma_{ss,i}^2$ is necessary. This quantity in turn depends on the level of investment of each of the individuals in the network (represented by the vector $\bm{k}$) and the  topology of the  underlying social interconnection graph $\mathcal G(L)$. The variances are the diagonal elements of the steady-state covariance matrix $\Sigma$ for the dynamics \eqref{sysm}, which can be computed by solving the matrix Lyapunov equation defined in Theorem \ref{th:lyap} (see also \cite{george1,poulakakisTAC} for related results on the consensus and drift-diffusion models respectively). 

\begin{thm}
\label{th:lyap}
For the stochastic dynamics \eqref{sysm} with social graph $\mathcal{G}(L)$, suppose there exists an $\mathcal{R}(L)$ such that $k_j>0$ for all $j\in\mathcal{R}(L)$. Then the system of stochastic differential equations \eqref{sysm} has steady-state mean $\displaystyle{\lim_{t\rightarrow \infty} E\left[\bm{\tilde x}(t)\right] =\bm{0}}$ and steady-state covariance matrix $\displaystyle{\Sigma=\lim_{t\rightarrow \infty} E\left[\bm{\tilde x}(t)^T \bm{\tilde x}(t)\right]}$ given by the solution to the Lyapunov equation:
\begin{equation}
 \displaystyle{(K_1 + K_2 L)\Sigma + \Sigma (K_1 + K_2 L)^T = SS^T}.
 \label{leq}
 \end{equation}
\end{thm}
\begin{proof}
The system \eqref{sysm} is a multivariate Ornstein-Uhlenbeck process with mean given by \cite{gardiner}$$E\left[ \bm{\tilde x}(t)\right]=\exp\left(Mt\right)E\left[\bm{\tilde x}(0)\right].$$ By Theorem \ref{th:necsuf}, $M$ is Hurwitz and hence $\displaystyle{\lim_{t\rightarrow \infty} E\left[\bm{\tilde x}(t)\right] =\bm{0}}$.

The covariance matrix of $\bm{\tilde x}$ is given by 
\begin{align*}
&E\left[(\bm{\tilde x}(t)-E\left[\bm{\tilde x}(t)\right])(\bm{\tilde x}(t)-E\left[\bm{\tilde x}(t)\right])^T\right] \\
&=\exp(Mt) E\left[\bm{\tilde x}(0) \bm{\tilde x} (0)^T\right] \exp(Mt) +\int\limits_0^t \exp(M(t-s))SS^T\exp(M^T(t-s))ds.
\end{align*}
Since $M$ is Hurwitz, the steady-state covariance matrix is given by $$\Sigma=\lim_{t\rightarrow \infty} E\left[\bm{\tilde x}(t) \bm{\tilde x}(t)^T\right]=\lim_{t\rightarrow \infty}\int\limits_0^t \exp(M(t-s))SS^T\exp(M^T(t-s))ds,$$ which as is shown in \cite{gardiner} is the solution to the Lyapunov equation:
\[
M\Sigma+\Sigma M^T + SS^T =0.
\]
\end{proof}

We leverage the results of Theorem \ref{th:lyap} to study the role of network topology in the evolutionary dynamics of collective migration in Sections  \ref{s:all}, \ref{s:limited}, and \ref{s:adaptive}.   We note that the theorem proves useful even in Section \ref{s:all}, where the underlying network topology is all-to-all, because in the general case the different individual investment strategies provide different edge weights in the social graph, making for a directed network topology.    

%%%%%%%%%%%%%%%%%%%%%%%%%%%%%%%%%%% 

\section{All-to-all case \label{s:all}}
\subsection{Model specialized to the all-to-all case}

As a first step in analyzing the evolutionary dynamics of the social migration model \eqref{sysm} with fitness \eqref{fitsol}, we consider the limit of all-to-all interconnection ($a_{ij}=\frac{1}{N-1}$ for all $i\neq j$ in \eqref{adj}) in a large population (labeled the {\em resident}  population with subscript $R$), with all individuals having a common level of investment $k_{R}>0$. This limit corresponds to the mean-field assumption used in \cite{torney}. By the law of large numbers, the average direction of population migration in the limit of large $N$ is the same as the desired migration direction $\mu$ after the decay of transients (i.e., in steady-state, $\displaystyle{\lim_{N\rightarrow\infty} L\bm{x} = \bm{x}-\mu \bm{1}}$). Substituting $L\bm{\tilde x}=\bm{\tilde x}$ in \eqref{sysm}, the dynamics of an individual in the population are given by
\begin{equation}
d\tilde x_{R} = -\left[ k_{R}^2 + (1-k_{R})^2\right]\tilde x_{R} \;dt + \sqrt{k_{R}^2 + \beta^2(1-k_{R})^3}\; dW.
\label{res}
\end{equation}

The corresponding steady-state variance of an individual's direction is given by 
\begin{equation}
\sigma_{ss,R}^2=\frac{k_{R}^2 + \beta^2(1-k_{R})^3}{2\left[k_{R}^2 + (1-k_{R})^2\right]},
\label{varss}
\end{equation}
with steady-state migration speed (performance) given by $\exp(-\sigma_{ss,R}^2/2)$. In Figure \ref{f:varybeta} we plot this steady-state migration speed as a function of investment $k_{R}$ for varying social noise term $\beta$. 
\begin{figure}[ht!]
  \begin{center}
    \includegraphics[width=.6\linewidth]{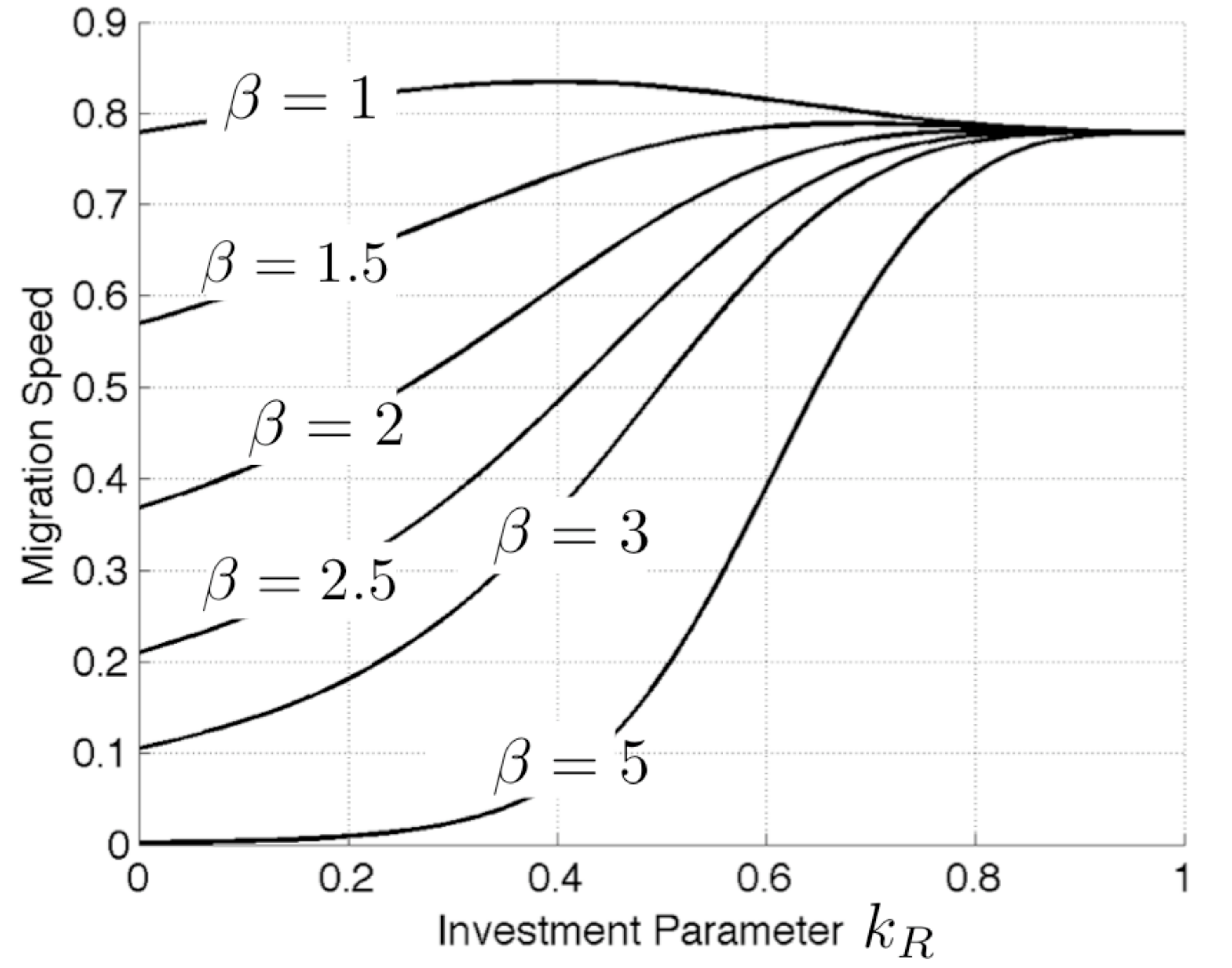}
    \caption{Steady-state migration speed as a function of resident population investment parameter $k_{R}$ and noise parameter $\beta$ for a large population with all-to-all interconnection. }
    \label{f:varybeta}
  \end{center}
\end{figure}
As defined in  \eqref{beta}, the parameter  $\beta$ reflects the strength of the noise from social interactions relative to the noise associated with the tracking process. In Figure \ref{f:varybeta} we see that the migration performance saturates at high levels of investment $k_{R}$, and remains low over greater $k_R$ ranges, for large $\beta$. We use $\beta>2$ in this work to model noisier social interactions relative to tracking (consistent with \cite{torney}); this provides an incentive for individuals to invest in the tracking process.

\subsection{Evolutionary adaptive dynamics \label{ss:ad}}

Now consider the evolution of strategies for such an all-to-all connected population.  A key part of any evolutionary algorithm is the computation of fitness of individuals in the population as a function of strategy distribution, model parameters, environmental conditions, and other such features. In certain cases (such as the all-to-all limit here), fitness can be analytically computed, which allows for an explicit calculation of the outcomes of the evolutionary process using tools from adaptive dynamics \cite{adyn1,adyn2,adyn3}. Adaptive dynamics are well-suited for studying the evolution of a continuous one-dimensional trait in a population undergoing small mutations.

Using \eqref{varss} and \eqref{fitsol}, the fitness of an individual in the resident population with dynamics \eqref{res} is given by 
\begin{equation}
F_R(k_R) =\exp\left( -\frac{k_R^2 + \beta^2(1-k_R)^3}{4(2k_R^2-2k_R+1)} -ck_R^2 \right).
\label{FR}
\end{equation}
Consider a small population of mutants with strategy $k_M$ interacting with each other and with all the residents. The mutants (owing to their small numbers) will experience the same social noise as the residents so their dynamics are 
\begin{equation}
d\tilde x_{M} = -\left[ k_{M}^2 + (1-k_{M})^2\right]\tilde x_{M} \;dt + \sqrt{k_{M}^2 + \beta^2(1-k_{R})(1-k_M)^2}\; dW.
\label{mut}
\end{equation}
Correspondingly, the fitness of individuals in the mutant population is given by
\begin{equation}
F_M(k_R,k_M) =\exp\left( -\frac{k_M^2 + \beta^2(1-k_R)(1-k_M)^2}{4(2k_M^2-2k_M+1)} -ck_M^2 \right).
\label{FM}
\end{equation}
The relative fitness of the mutant strategy in the environment of the resident is known as the differential fitness and is given by
\begin{equation}
S(k_R,k_M) =F_M(k_R,k_M) - F_R(k_R). 
\label{Sm}
\end{equation}

For a given resident strategy $k_R$, the values of $k_M$ that result in $S>0$ correspond to the mutant strategies that when rare can invade the established resident population. Further, a study of the selection landscape $S$ can help us predict conditions for an evolutionarily stable monomorphic population (all individuals having the same strategy) and conditions for evolutionary branching into subpopulations of leaders ($k_i\approx 1$) and followers ($k_i \approx 0$), as a function of the cost $c$ associated with strategy investment.

The evolutionary dynamics of the resident strategy $k_R$ are given by
\begin{equation}
\frac{dk_R}{d\tau} = \gamma \left.\frac{\partial S}{\partial k_M}\right|_{k_M=k_R}=: \gamma\; g(k_R),
\label{addt}
\end{equation}
where $g(k_R)$ is the selection gradient  and $\gamma$ is a positive scalar constant.  We note that the timescale $\tau$ associated with \eqref{addt} corresponds to slow evolutionary time and is different from the fast timescale $t$ associated with the stochastic migrations dynamics \eqref{sysm} and \eqref{res}. For the differential fitness $S$ defined in \eqref{Sm}, \eqref{FR} and \eqref{FM}, singular strategies $k_*$ corresponding to equilibria of \eqref{addt} (defined as $g(k_*)=0$)  are given by the solutions to the following  equation (details in \ref{s:app}): 
\begin{equation}
k_*\left[\beta^2(1-k_*)-1\right](k_*-1)+4ck_*(2k_*^2-2k_*+1)^2=0.
\label{sing}
\end{equation}

A singular strategy $k_*$ is known as a Convergent Stable Strategy (CSS) if it is locally asymptotically stable for the dynamics \eqref{addt}, i.e., if
\begin{equation}
\left.\frac{\partial g}{\partial k_R}\right|_{k_R=k_*}<0. 
\label{css}
\end{equation}
A CSS strategy $k_*$ can be either locally evolutionary stable (local ESS) for the population, or it can be a branching point for the population. The condition for it to be a branching point is  
\begin{equation}
\left.\frac{\partial^2 S} { \partial k_M^2} \right|_{k_M=k_R=k_*}>0.
\label{branchcon}
\end{equation}

For $S$ defined in \eqref{Sm}, the branching condition \eqref{branchcon} evaluates to (details in \ref{s:app})
\begin{equation}
0<k_*<\frac{5-\sqrt{7}}{6}\approx 0.3923.
\label{branch}
\end{equation}
The branching condition  \eqref{branch} is exactly the same as that obtained in \cite{torney}. Parameters $\beta$ and $c$ of the dynamics  that yield CSS singular strategies $k_*\in\left(0,\frac{5-\sqrt 7}{6}\right)$ result in populations with distinct leader and follower groups via evolutionary branching. 

\subsection{Results}

\begin{figure}[ht!]
  \begin{center}
    \includegraphics[width=.6\linewidth]{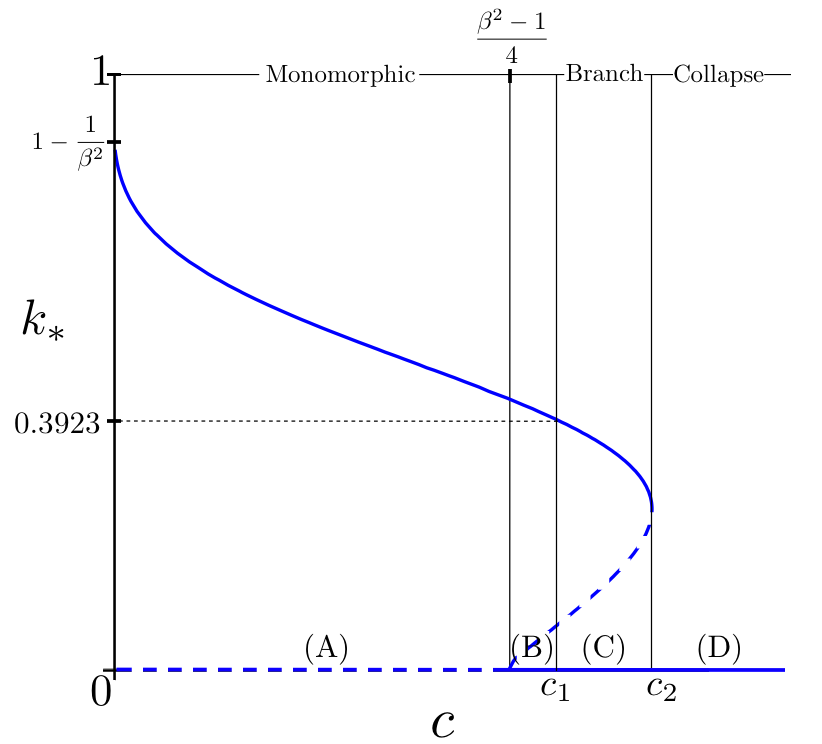}
    \caption{Evolutionary singular strategies $k_*$ as a function of cost parameter $c$. The two sets of singular strategies defined by \eqref{sing} are plotted in blue. One set corresponds to $k_*=0$ and the other corresponds to the curve given by the equation $c=\frac{(1-k_*)[\beta^2(1-k_*)-1]}{4(2k_*^2-2k_*+1)^2}$. Solid curves are CSS strategies, and dashed curves are unstable singular strategies. The regions marked (A)-(D) correspond to the descriptions in the text. Analytical derivations for the cost parameters $\frac{\beta^2-1}{4}$, $c_1$ and $c_2$ that divide the regions are given in \ref{s:app}.}
     \label{f:nice}
  \end{center}
\end{figure}

The bifurcation diagram in Figure \ref{f:nice} summarizes the singular strategy condition \eqref{sing}, CSS condition \eqref{css} and branching condition \eqref{branch} to obtain four distinct sets of evolutionary outcomes (in ranges A, B, C and D) for our model for increasing cost parameter $c$ (calculations in \ref{s:app}):
\begin{enumerate}
\renewcommand{\labelenumi}{(\Alph{enumi})}
\item Monomorphic population: For $0<c<\frac{\beta^2-1}{4}$ there exists only one CSS strategy, and it is evolutionarily stable (since $k_*>\frac{5-\sqrt 7}{6}$), resulting in a monomorphic population with strategy $k_*$. The other singular strategy corresponding to the fully social strategy $k_*=0$ is not convergent stable.  
\item Two local CSS's that are each evolutionarily stable: For $\frac{\beta^2-1}{4}<c<c_1$ there exist two convergent stable strategies, one of which is the fully social strategy $k_*=0$. Both singular strategies are locally evolutionarily stable. $c_1$ is defined implicitly as the solution to the equation $\left.g(\frac{5-\sqrt{7}}{6})\right|_{c=c_1}=0$.   
\item Branching: For $c_1<c<c_2$ the convergent stable interior ($k_*\in(0,1)$) singular strategy satisfies the branching condition \eqref{branch} resulting in the population splitting into distinct leader and follower groups. 
\item Collapse of Migration: For the high cost scenario $c>c_2$, the only singular strategy that exists is the convergent stable fully social strategy $k_*=0$ and the population does not develop any migration ability. 
\end{enumerate}

These four cases constitute a comprehensive picture of the evolutionary dynamics of the migration model \eqref{sysm} in the all-to-all limit, and encompass key features of the branching calculations in \cite{torney} and evolutionary simulations in \cite{guttal}. The existence of two locally evolutionary stable attractors in case (B) above implies that the evolutionary dynamics can potentially yield two-strategy outcomes in the population without evolutionary branching (case (C)). 

A hysteretic effect associated with restoring population migration ability once destroyed is apparent in Figure \ref{f:nice}. In particular, once migration ability in the population is lost for high cost parameter $c>c_2$, the cost parameter needs to be reduced below the level $\frac{\beta^2-1}{4}$ (i.e., below $c_1<c_2$) for migration ability to be regained. This compares to the simulations in \cite{guttal,guttal2} where agent-based models are used to study the effect of habitat fragmentation on the evolution of migration (see also \cite{Fagan2012}). In these simulations, the authors study the impacts of progressively more fragmented habitats on migratory outcomes, and show that once migration ability is lost for a threshold level of fragmentation, much greater habitat recovery is necessary to restore lost migration ability (a hysteretic effect). Higher levels of habitat fragmentation are comparable to higher cost  $c$. 

\begin{figure}[ht!]
  \begin{center}
    \includegraphics[width=1\linewidth]{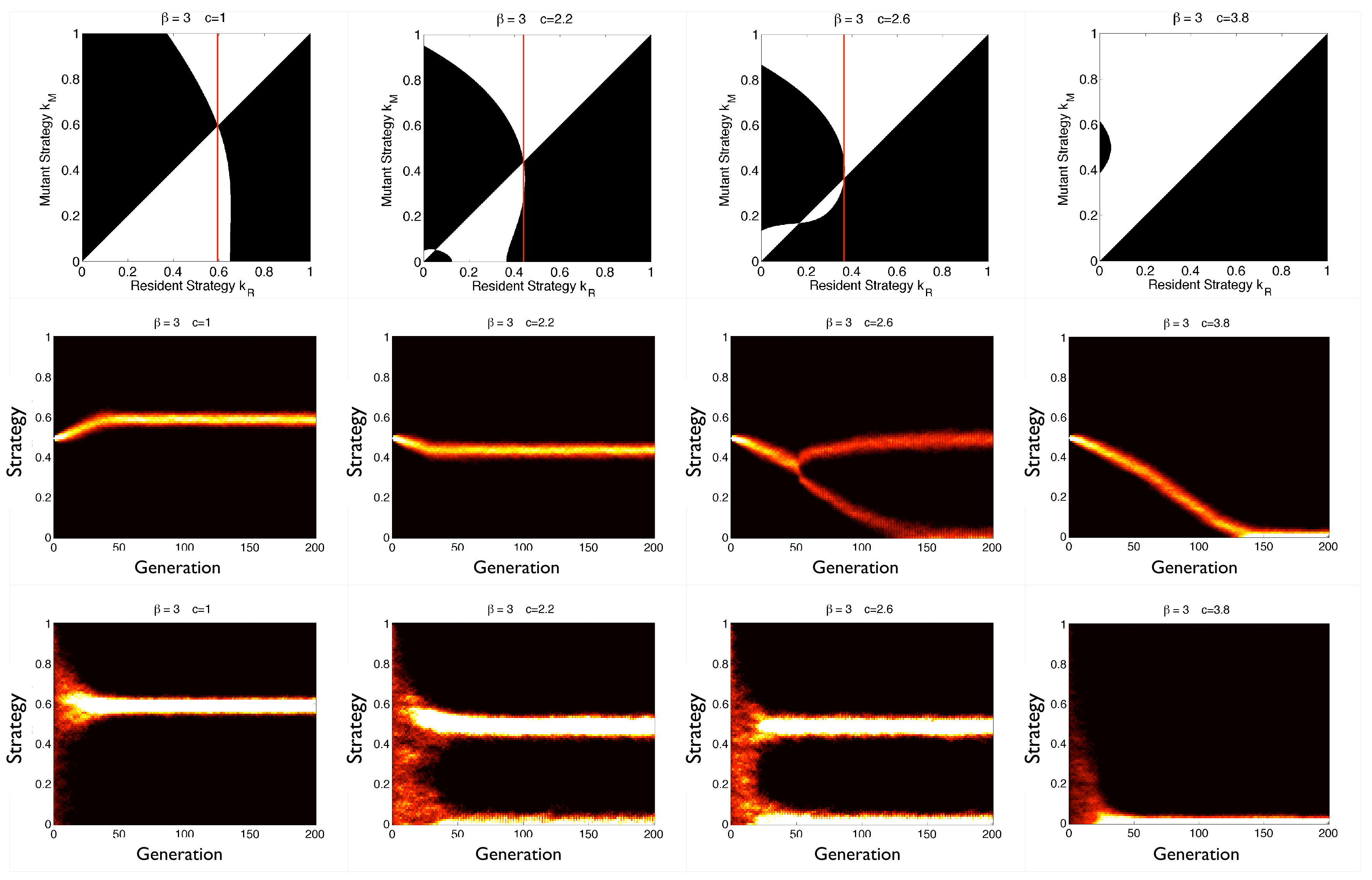}
    \caption{Evolutionary dynamics for the migration model with all-to-all interconnection, noise parameter $\beta=3$, and with increasing cost parameter $c$ from left to right:  $c=1, 2.2, 2.6, 3.8$. The top row shows the pairwise invasibility plots for $k_M$ vs.\ $k_R$, each from 0 to 1. Black regions correspond to differential fitness $S(k_R,k_M)>0$ (mutants can invade) and white regions to $S<0$. The red vertical lines pass through convergent stable interior strategies $k_R=k_*$. The middle row of plots are evolutionary simulations starting with a monomorphic population with strategy $k=0.5$; hot colors correspond to high population density. The bottom row of plots are also evolutionary simulations, but having an initial population with a uniformly randomly distributed strategy $k\in[0,1]$. We use $N=2000$ individuals for these simulations.   The threshold costs as defined in Figure \ref{f:nice}  for $\beta=3$ are $c_1=2.48$ and $c_2=2.77$. }
    \label{f:fire}
  \end{center}
\end{figure}

In Figure \ref{f:fire} we show the pairwise invasibility plots (PIPs) \cite{adyn1,adyn2,adyn3} of the differential fitness $S (k_R , k_M )$ for increasing cost $c$ to illustrate the four sets of outcomes described above. These plots show the sign of $S$ as a function of the resident and mutant population strategies. Dark regions correspond to differential fitness $S > 0$ and allow mutant invasions; white regions correspond to $S < 0$ and prohibit mutant invasions. The population resides on the diagonal, and is monomorphic when a negative differential fitness is associated with the red vertical line through the singular strategy $k_*$. The PIPs in Figure \ref{f:fire} show conditions for an initially monomorphic population for low values of cost parameter $c$ and population branching for intermediate values of $c$. For high values of $c$ we see conditions that prevent individuals from developing any significant investment and therefore any significant migration ability. 

We confirm our predictions from the adaptive dynamics analysis by running evolutionary simulations in the case of a monomorphic initial condition and also a uniformly randomly distributed initial condition as shown in Figure \ref{f:fire} (bottom two rows). These simulations comprise the roulette-wheel selection \cite{garef} and small mutation operations on each generation of a population of $N=2000$ individuals with all-to-all social graph, dynamics \eqref{sysm}, and fitness \eqref{fitsol}.   The fitnesses are computed for each generation using equation (\ref{leq}) of Theorem \ref{th:lyap}.

The columns in Figure \ref{f:fire} from left to right correspond to the cases (A)-(D) respectively. In case (A), both initial conditions result in a monomorphic evolutionary outcome. In case (B), the polymorphic solution for the evolutionary simulation with random initial conditions (Column 2, last Row) is a consequence of the stability of the $k_*=0$ singular strategy, and not  a consequence of branching, as is the case in  (C). Case (D) corresponds to the collapse of migration with all individuals having an insignificant level of investment. 

The analysis in this section shows the range of evolutionary outcomes for the migration model with all-to-all social interconnection. We are particularly motivated by conditions that result in the branching of the population into invested leaders and social followers. In the following section we study the role that limited social interconnection topology plays in the emergence of this evolutionary branching. 

%%%%%%%%%%%%%%%%%%%%%%%%%%%%%%%%%%% 

\section{Limited interconnections case \label{s:limited}}

While the all-to-all topology assumption of Section \ref{s:all} allows for a detailed analysis of the evolutionary dynamics in a large population, it is unrealistic for most biological and decentralized artificial systems. In this section, we relax the all-to-all assumption and study the migration model \eqref{sysm} with underlying limited interconnection. 
 
 The Lyapunov equation \eqref{leq} in Theorem \ref{th:lyap} allows us to compute (without simulation) the migratory performance (and correspondingly fitness \eqref{fitsol}) of individuals since the diagonal terms of the steady-state covariance matrix $\Sigma$ are the individual steady-state variances $\sigma_{ss,i}^2$. In this section, we use evolutionary simulations based on fast timescale fitness calculations from \eqref{leq} to study the role that graph connectivity plays in the evolution of branching.

We focus on three classes of social graph topologies of which one is ordered (ring lattice) and two are random (undirected and directed). In each class, a single parameter controls the level of connectivity of the graph. The three classes of graphs used are listed below: 
\begin{itemize}
\item Undirected Ring Lattice:  A graph with $N$ nodes, each connected to $K$ nearest neighbors, $K/2$ on each side, for $K$  even. An undirected edge exists between nodes $i$ and $j$ if and only if $0<\text{min}\left\{|i-j|,N-|i-j|\right\}\leq K/2$. The graph is connected for $K\geq 2$. 
\item Random Undirected (Erd\H{o}s-R\'enyi) \cite{random,newman}: Undirected graph with $N$ nodes. Every edge in the graph exists randomly with a uniform probability  $p$. The expected number of neighbors of a node is $E[K]=Np$ for large $N$. The graph is almost surely connected if  $p>\ln(N)/N$ or equivalently $E[K]>\ln(N)$. 
\item Random Directed \cite{randd1,randd2}: Directed graph with $N$ nodes. Each node has a probability $p$ of having a directed link to every other node in the network. The expected number of neighbors of a node is $E[K]=Np$ for large $N$. 
\end{itemize}

 \begin{figure} [ht!]
 \begin{center}
\includegraphics[width=0.8\linewidth]{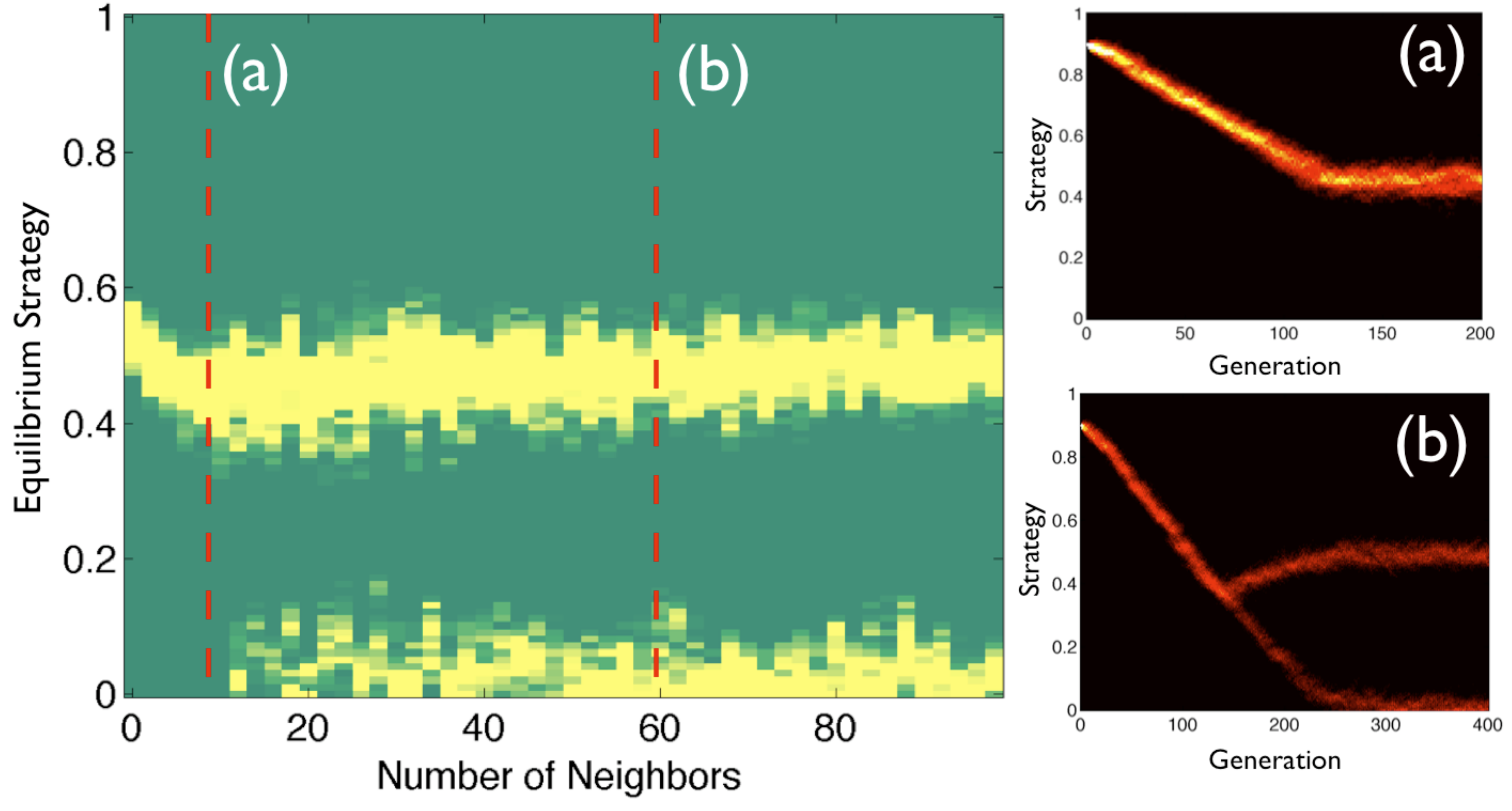}
\caption{Effect of number of neighbors $K$ on the evolutionary outcomes of the migration model. The left plot shows the equilibrium strategy distribution as a function of number of nearest neighbors for the ring lattice graph model with $N=400$ nodes and parameters $\beta=3$ and $c=2.6$; bright colors correspond to higher population density. The two plots on the right labeled (a) and (b) are evolutionary simulations, the steady-state conditions in these plots correspond to the red dashed slices in the left plot. The  two-strategy equilibrium exists only once the graph connectivity exceeds a threshold number of neighbors ($K \approx 15$ for parameters chosen here). }
\label{f:topodemo}
\end{center}
\end{figure} 

For each class of topologies, the parameters $K$ and $p$ allow us to explore a range of connectivities; for $K=p=0$, the social graphs are fully disconnected and individuals must resort to solitary migration with a monomorphic optimal strategy \eqref{optsol}. For $K\rightarrow  N-1$ (for the ring lattice) and $p\rightarrow 1$ (for the random graphs), the social graph is fully connected, resulting in the leader and follower  evolutionary equilibrium for certain parameter choices as discussed in Section \ref{s:all}. Between these two connectivity extremes, intuition suggests that an intermediate level of limited connectivity can provide adequate information flow in the network for followers to leverage the investments made by leaders, thereby resulting in the two distinct populations. In Figures \ref{f:topodemo} and \ref{f:topo9} we confirm this intuition by showing that the transition from a monomorphic solution, to a branched evolutionary solution as a function of topology (parameterized by $K$ and $p$), occurs at an intermediate threshold level of connectivity.   

For the simulation in  Figure \ref{f:topodemo} we use the ring lattice topology for the social graph with $N=400$ individuals and choose a set of parameters for which the all-to-all social graph is known to have a two-strategy solution, $\beta=3$ and $c=2.6$ (see Figure \ref{f:fire}, column 3). For a range of values of the number of neighbors $K=2,4,\cdots,100$, we compute the fitness of individuals using \eqref{leq} and \eqref{fitsol} and evolve the strategies of the populations to an evolutionary steady state. This steady-state strategy distribution is plotted as a function of number of neighbors $K$ in Figure \ref{f:topodemo}. We see that there exists an intermediate threshold for social connectivity (given by $K\approx 15$) that allows for adequate information flow in the network to result in evolutionary branching. That is, social graphs that are much sparser (fewer edges) than the all-to-all case analyzed in Section \ref{s:all} can yield leader and follower evolutionary outcomes. 

In Figure \ref{f:topo9} we present simulations similar to those described for Figure \ref{f:topodemo} above, for all three classes of social graph topologies (ring lattice, undirected random, directed random) and for $N=200,400, 600$. In each case, we see a minimum connectivity threshold for evolutionary branching that is higher than the minimum threshold for the social graph to be connected. We also see that this threshold for branching is not affected by increasing population size $N$ in the range of $N$ considered; i.e., the minimum fraction of the population that each node must be connected to for branching decreases with population size. Further, the location of the threshold is dependent on the class of graph being considered; the two classes of random graphs have lower thresholds than the ordered ring lattice. 

\begin{figure} [ht!]
 \begin{center}
\includegraphics[width=1\linewidth]{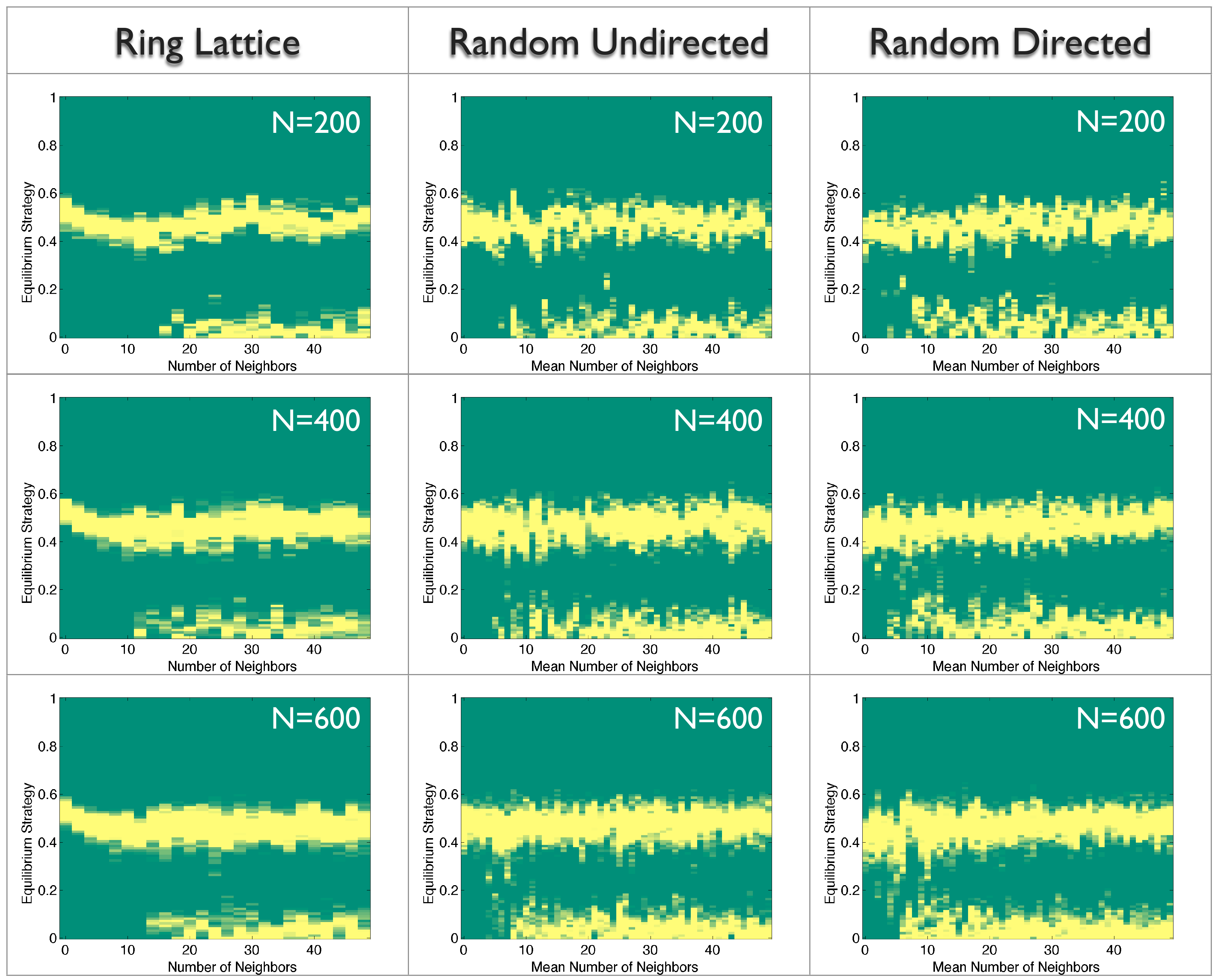}
\caption{Evolutionary equilibria as a function of social graph topology for ring lattice, random undirected and random directed graph models with parameters $\beta=3$ and $c=2.6$; number of nodes $N$ are shown on each plot. The minimum number of neighbors (or mean number for the random graphs)  is independent of the number of nodes $N$ in these plots, $\approx 15$ for the ring lattice, $\approx 9$ for the random undirected graphs and $\approx 8$ for the random directed graphs.  }
\label{f:topo9}
\end{center}
\end{figure} 

The simulations in this section illustrate the significant effect of limited social graph connectivity on the evolution of  branching in migration. In particular we show that social connectivity above threshold levels can yield two-strategy outcomes (leaders and followers) in the evolutionary dynamics. The thresholds depend on the classes of social graph topology being considered.  Determining the analytical minimum connectivity threshold as a function of parameters $\beta$, $c$ and class of graph, is a topic of interest. A natural next step is to look at other classes of graph topologies such as spatially embedded graphs with a topological metric on connectivity (such as in \cite{star1}), and classes of small-world graphs parameterized by a single rewiring parameter \cite{wattsstrogatz, newman}. The ring lattice and undirected random graphs considered here are two extreme limits of the Watts-Strogatz \cite{wattsstrogatz} small world connectivity model.

%%%%%%%%%%%%%%%%%%%%%%%%%%%%%%%%%%% 

\section{Adaptive nodes \label{s:adaptive}}
Our analysis of the collective migration model \eqref{sysm} thus far has focused on the evolutionary perspective,  considering the dynamics of networks with large numbers of nodes. As discussed in Section \ref{s:intro}, the stochastic dynamics \eqref{sysm} can be interpreted more generally as collective tracking or learning dynamics in multi-agent systems of small or large size. In this section we shift focus and consider the model \eqref{sysm} from an adaptive perspective on small dynamic networks. This analysis is motivated in part by questions about leadership, task assignment, learning and robust adaptive behavior in multi-agent robotic systems.

We consider a simple model of greedy local optimization by nodes on a graph, which yields individual adaptation of investments $k_i$.  For this model we show bifurcations as a function of cost that yield leader-follower emergent behavior as equilibria of the adaptive process. We also illustrate the critical role played by graph topology in determining the {\em location } of leaders in the network of adaptive nodes using several examples.

Consider a system of interconnected agents with fast timescale tracking dynamics given by \eqref{sysm}. Further, suppose that each agent seeks to maximize its local utility function by adapting its investment parameter $k_i$. We assume that the utility function for each agent $U_i$ is given by the fitness function \eqref{fitsol}, 
\begin{equation}
U_i = \exp\left(\frac{-\sigma_{ss,i}^2}{2} \right) \exp\left( -ck_i^2\right).
\label{U}
\end{equation}
In this setting, the utility function for a focal agent $i$ depends on the agent's investment $k_i$, as well as the investments of other nodes in the network; we assume that an agent can measure its own utility, but not the investments of other agents. Each agent modifies its investments $k_i$ on a slow timescale by  climbing the gradient of its local utility $U_i$ to reach its local maximum,
\begin{equation}
\dot{k}_i = \frac{dk_i}{dt} = \frac{\partial U_i}{\partial k_i},\;\;i=\left\{ 1,\cdots,N \right\}.
\label{dk}
\end{equation}  
Our goal is to study the outcomes of this  adaptive process by computing equilibria of the dynamics defined by \eqref{dk}, \eqref{U} and \eqref{sysm}, and studying their bifurcations. 

 \subsection{$N=2$ case}
We first look at the simplest case of the dynamics with $N=2$ nodes, with an underlying all-to-all graph. In this case, the steady-state covariance matrix $\Sigma$ (from \eqref{leq}) can be computed analytically \cite{gardiner} as
\begin{equation}
\Sigma =\frac{(\text{Det } M)SS^T + \left[-M+(\text{Tr }M)\right]SS^T\left[-M+(\text{Tr }M)\right]^T}{-2(\text {Det } M)(\text{Tr } M)},
\label{gar}
\end{equation}
where matrices $M$ and $S$ are defined in \eqref{sysm} and \eqref{nf}. For each pair $\{i,j\}=\{1,2\},\{2,1\}$, the diagonal elements of $\Sigma$ from \eqref{gar} are given by
\begin{align}
\label{sij}
&\sigma_{ss,i}^2 = \frac{f_{ij} \left(k_i-1\right){}^4+f_{ji} \left(2 k_j^2-2 k_j+1\right){}^2+f_{ji} g_{ij} }{4 g_{ij} \left( k_i^2 + k_j^2 -k_i-k_j + 1\right)},\\
&\text{where } f_{ij} = \left(k_j^2-\beta ^2 \left(k_i-1\right) \left(k_j-1\right){}^2\right),\nonumber \\
&\text{and    } \;\;\;g_{ij} = \left(3 k_j^2-2 k_j+1\right) k_i^2-2 k_j^2 k_i+k_j^2\nonumber.
\end{align}

Substituting \eqref{sij} in \eqref{dk} and \eqref{U}, we compute the equilibria $\bm{k}_{eq}$ of the dynamics \eqref{dk} and \eqref{sysm} and their stability as a function of increasing cost parameter $c$. Analytical expressions of the equilibria are complicated; we illustrate the equilibria  for $\beta=3$ in Figure \ref{f:two}. For low cost, both individuals make a significant equal investment corresponding to the symmetric equilibrium $k_{eq,1}=k_{eq,2}\gg 0$. As cost increases, the level of this equilibrium investment decreases and eventually a pair of stable leader-follower equilibria appear via two saddle-node bifurcations. At this cost level, it is quite interesting to note that both the symmetric solution (monomorphic population) and the two-strategy solution (branching populations) coexist as stable solutions.  However, as the cost increases further, the symmetric stable equilibrium loses stability in a subcritical pitchfork bifurcation, leaving the leader-follower pair of stable equilibria and an unstable symmetric saddle equilibrium.    A phase portrait for the slow timescale dynamics (strategies $k_1$ and $k_2$) is shown in Figure~\ref{f:two} for each of three values of cost $c$ to illustrate the different sets of equilibrium solutions.
\begin{figure} [ht!]
 \begin{center}
\includegraphics[width=0.9\linewidth]{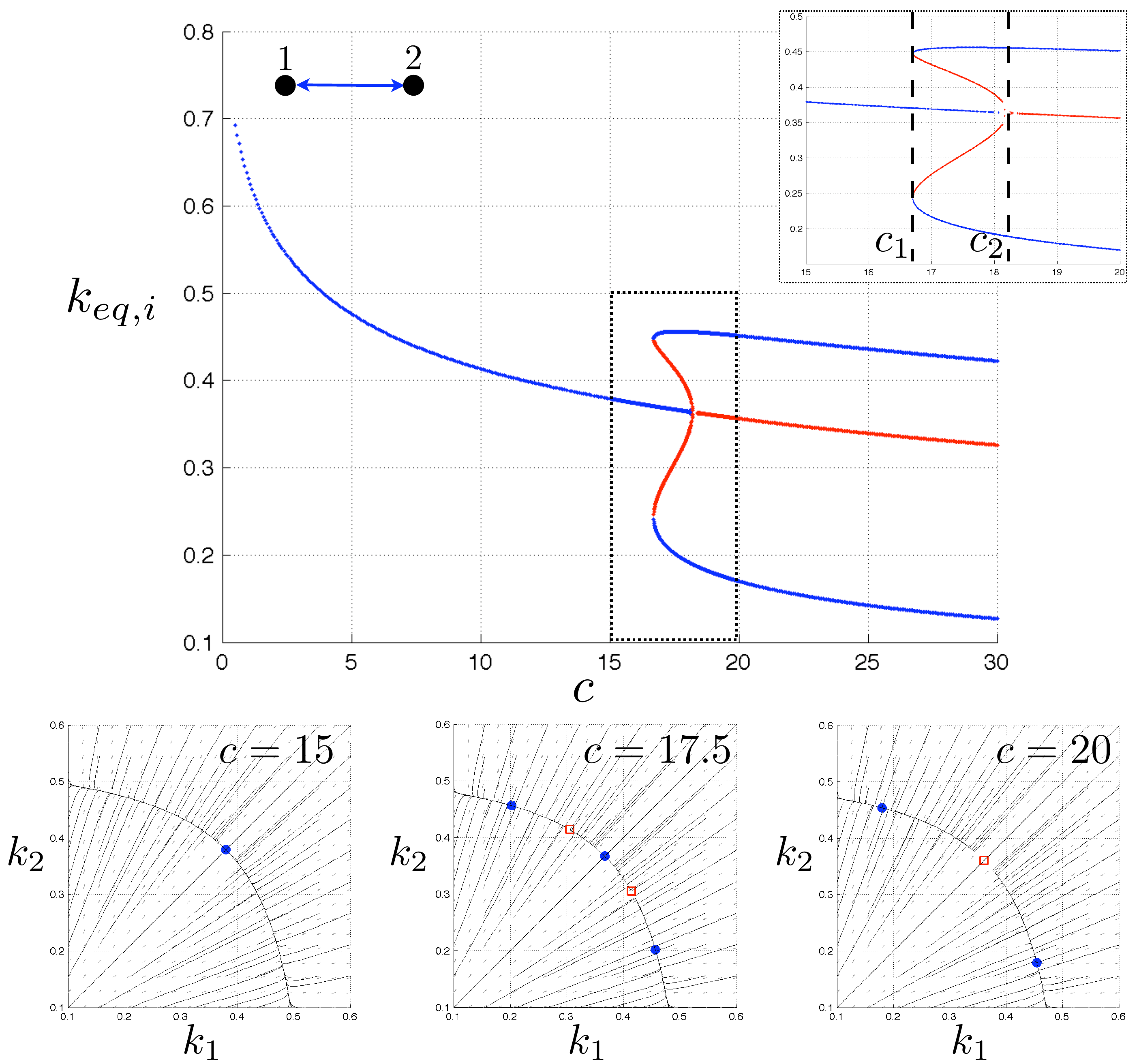}
\caption{Bifurcations for the adaptive node dynamics \eqref{dk} with $N=2$ nodes, an underlying all-to-all social graph, and noise parameter $\beta=3$. The top plot shows the two components of $\bm{k}_{eq}$ (equilibria of the dynamics \eqref{dk} and \eqref{sysm}) as a function of the cost parameter $c$. Stable sinks are marked blue and unstable saddles are marked red. The inset shows a zoomed in view of the region with $15\leq c \leq20$ marked in the dotted square. The dashed lines in the inset $c_1\approx 16.7$ and $c_2 \approx 18.2$ denote the saddle-node and pitchfork bifurcation points respectively. The row of bottom plots are phase portraits for the slow timescale dynamics with parameter $c$ as indicated; the circles are stable sinks and the squares are saddles. These plots remain qualitatively the same for different values of $\beta>2$; the bifurcation points $c_1$ and $c_2$ move further to the right for higher $\beta$.    }
\label{f:two}
\end{center}
\end{figure}

\subsection{$N>2$ all-to-all}
\begin{figure} [ht!]
 \begin{center}
\includegraphics[width=0.9\linewidth]{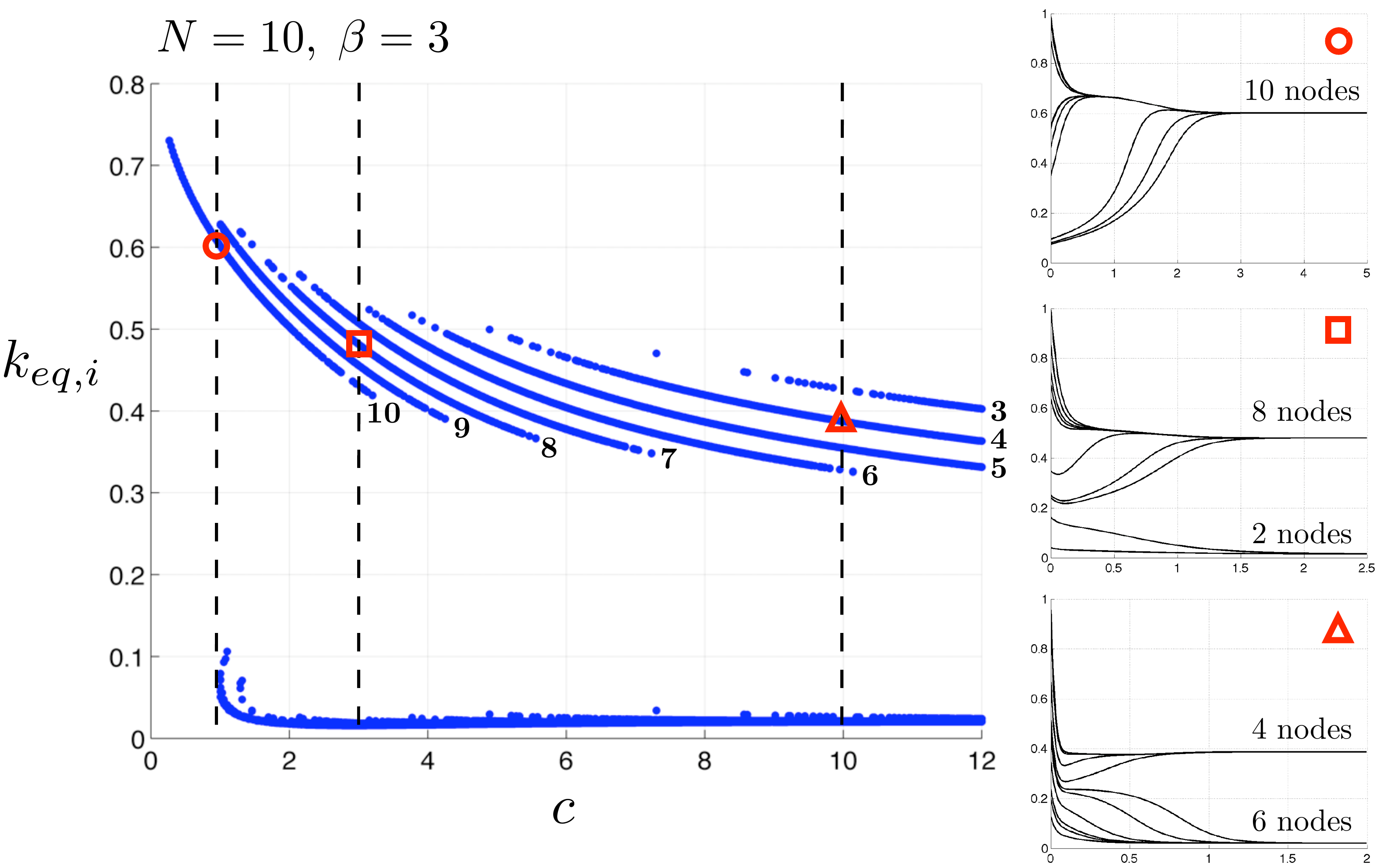}
\caption{Bifurcations for the adaptive node dynamics \eqref{dk} and \eqref{sysm} with $N=10$ nodes, an underlying all-to-all social graph, and noise parameter $\beta=3$. The left plot shows stable equilibria of the dynamics as a function of $c$; plots of $k_i$ vs.\ time for the points marked with the circle ($c=1$), square ($c=3$) and triangle ($c=10$) are shown on the right. The labels on the left plot indicate the number of leaders in each branch of stable solutions. It can be observed that bifurcations yield fewer leaders for increasing cost.}
\label{f:ten}
\end{center}
\end{figure} 

For larger networks with underlying all-to-all graph and dynamics \eqref{dk} and \eqref{sysm}, bifurcations in cost $c$ yield generalizations of the bifurcations observed in the $N=2$ case.  Figure~\ref{f:ten} illustrates these bifurcations in the case of $N=10$.   The plot on the left shows the equilibrium values of the strategies $k_i$ as a function of $c$.  For a given value of $c$ there are multiple co-existing stable leader-follower solutions, each solution differing by the number of leaders (and therefore the number of followers).   For example, at $c=10$ there are distinct stable solutions with 3, 4, 5 and 6 leaders; the adaptation of $k_i$ as a function of time is shown to the right at the bottom (marked by a triangle) for initial conditions that yields 4 leaders and 6 followers.  Other examples are shown in the case $c=1$ and $c=3$.   It can be observed from Figure~\ref{f:ten} that the fraction of nodes in the leader populations decreases with increasing cost parameter $c$ as a consequence of several bifurcations in the dynamics. 

We note that although the underlying social graph is undirected and all-to-all, the adapted equilibrium solutions correspond to  strongly directed  graphs.  In particular, the edges in the graph from leaders to their neighbors have small weights, whereas the edges from followers to their neighbors have large weights.

\subsection{Star topology and other topologies}
Above we show how the cost parameter $c$ influences the number of leaders and followers that emerge from the adaptive dynamics in the case of an underlying all-to-all social graph.   Here we study how the topology of limited social interconnections influences the emergence and location of leaders and followers.   We illustrate this for a social graph defined by a star in Figure \ref{f:star}; the star graph has a significant asymmetry since the central node can sense all  $N-1$ fringe nodes, whereas the fringe nodes can each only sense the central node.   

From the left plot in Figure~\ref{f:star}, it can be seen that at low values of cost $c$, the fringe nodes of the star invest strongly in the external signal as leaders while the central node leverages these neighbors as a follower with small investment.  In this case, the fringe nodes pay almost no attention to social cues from the central node  while the central node relies almost exclusively on social cues from the fringe nodes (an exploding star graph).  At intermediate cost $c$, as shown in the middle plot, all nodes make similar investments (monomorphic population).   At high cost $c$, as shown in the right plot, the central node adapts to become the leader and ignores the fringe nodes, while all the fringe nodes leverage the central node's investment as followers (an imploding star graph).
\begin{figure} [ht!]
 \begin{center}
\includegraphics[width=1\linewidth]{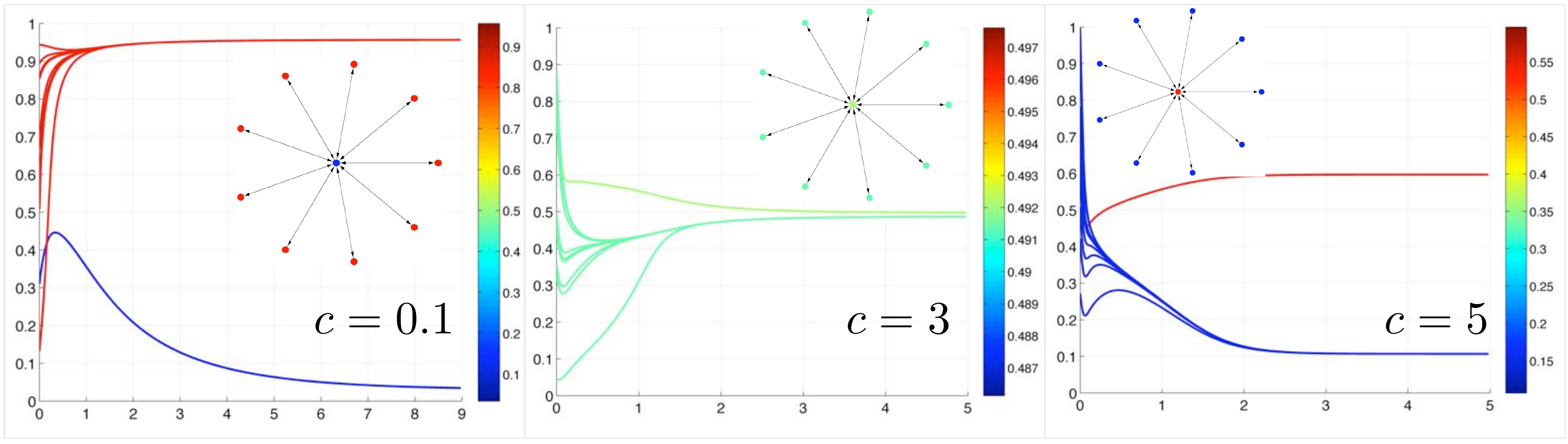}
\caption{Role of topology and cost parameter $c$ for the adaptive node dynamics \eqref{dk} and \eqref{sysm} with $N=10$, $\beta=3$, and underlying social graph given by the  star graph.  In the star graph,  the central node can sense all the fringe nodes but the fringe nodes can each only sense the central node. Each plot shows a simulation of the ten strategies $k_i$ versus time; the parameter $c$ is indicated on each plot and can be seen to increase from  left to right.   The color-scale corresponds to the magnitude of equilibrium investment $k_{eq,i}$ (higher values are hotter colors); these colors are superimposed on the picture of the star graph for each of the three simulations shown so that nodes with hotter colors can be observed to be the invested leaders.   }
\label{f:star}
\end{center}
\end{figure} 

We show equilibrium outcomes for three more underlying social graph topologies in Figure \ref{f:color}. For the undirected ring lattice shown on the left, and parameter values indicated, alternate nodes adapt to become leaders (red) while the others become followers (blue).  This is not surprising since an agent need not invest if both of its immediate neighbors are  heavily invested.  
For more complicated topologies, the  connection between topology and location of emergent leaders is more challenging to interpret. Making the connection requires the development of a graph and investment dependent metric to rank nodes for their leadership potential, such as the {\it information centrality} metric used to rank certainty of nodes in networks of stochastic evidence accumulators  \cite{poulakakisTAC}. 

\begin{figure} [ht!]
 \begin{center}
\includegraphics[width=.8\linewidth]{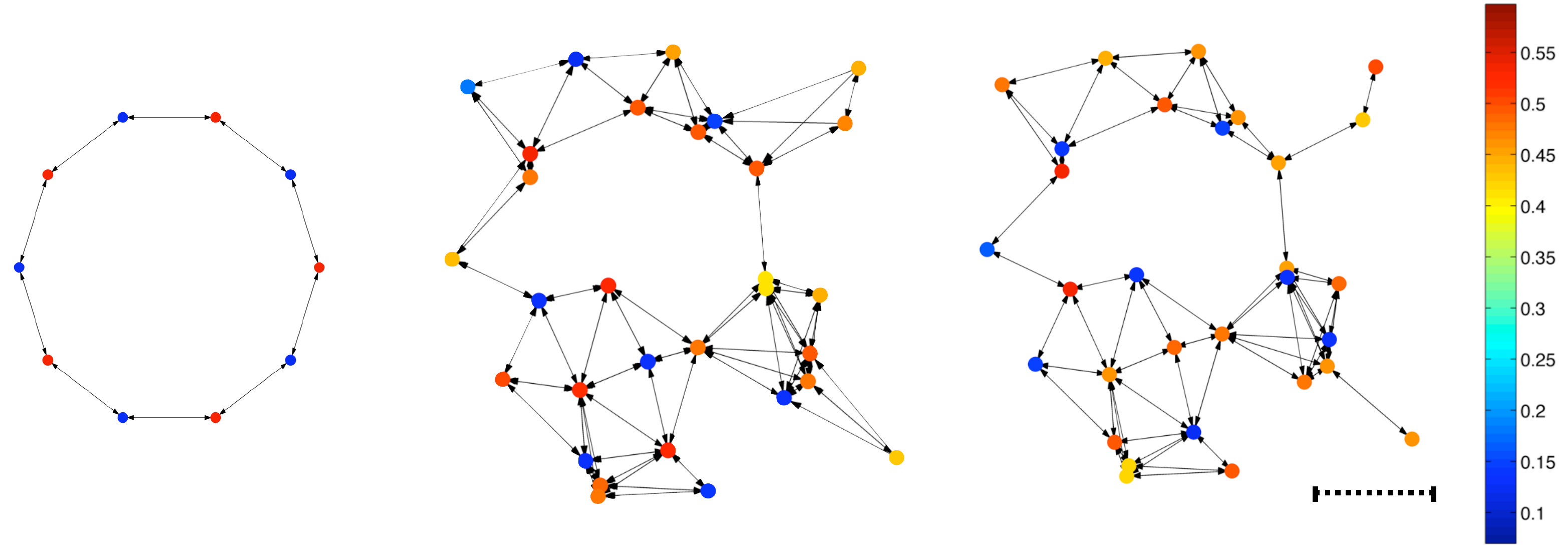}
\caption{Role of topology in determining locations of leaders on the social graph. Parameters for all three plots are $\beta=3$ and $c=4$. The colors of nodes on each plot correspond to equilibrium investments $k_i$ for the dynamics \eqref{dk} and \eqref{sysm} with magnitudes indicated in the color-bar. The left graph is an undirected star with $N=10$ nodes. The two right plots show a random spatial embedding of nodes with two different interconnection models. In the middle plot, each node is connected to its three nearest neighbors (topological distance) and in the right-most plot, each node is connected to all neighbors within a fixed radius given by the dashed line drawn (metric distance).    }
\label{f:color}
\end{center}
\end{figure}

%%%%%%%%%%%%%%%%%%%%%%%%%%%%%%%%%%% 

\section{Final remarks \label{s:final}}
The study of leadership in dynamic networks has received significant attention in both biology and multi-agent robotics. In biology there has been great interest in  determining conditions for the stable evolution of leadership in systems of socially interacting agents where there is incentive for individuals to  ``free-ride'' on the investment of others. In networked robotic systems, the leader-follower paradigm has been studied in a variety of contexts as a tool to design control protocols that achieve desired group performance. 

In this work, we derive an analytically tractable adaptive dynamic network model of collective migration in which investment (leadership) is costly and social interactions are cheap, and we use this model to  study the role of the social interconnection graph in the evolution of leadership.    We analyze bifurcations in the dynamics as a function of the cost of investment by leveraging a two-timescale dynamic with fast timescale migration dynamics and slow timescale adaptation of  investments and social graph edge weights.     In the all-to-all limit of the evolutionary model,  
our complete bifurcation analysis explains phenomena such as the observed hysteretic effect associated with recovery of migration in fragmented environments.   In the general case of social interactions with limited connectivity, we prove necessary and sufficient conditions on the graph and investments such that  the fitness of individuals in the stochastic migration model can be computed from a Lyapunov equation.  We use this solution to show a minimum connectivity threshold in random and ordered networks above which there is evolutionary branching into leader and follower populations.    

Finally, we use our evolutionary collective migration model and fitness equation to study smaller networks inspired by collective robotic systems; we replace the replication and mutation dynamics of the slow evolutionary timescale with dynamics in which agents adapt their individual investment and social interaction strategy according to a greedy optimization of their own utility function (fitness).   We study bifurcations in these dynamics as a function of investment cost and show how the topology of the underlying network graph plays a critical role in determining the emergence and location of leaders in the adaptive system. 
 
Further work is needed to derive rigorous minimum bounds on connectivity for branching of populations into leader and follower groups and generalization of our nonlinear dynamics analysis of adaptive networks using formal graph metrics.   It is also of great interest to consider relaxing modeling assumptions and  generalizing the form of the dynamics.  Another important avenue for future work is to investigate where the top-down engineering design approach for group performance meets the bottom-up adaptive approach taken in this paper, i.e., to explain how  group performance emerges from the adaptive network dynamic model and how we might leverage the bottom-up approach for design of high performing and resilient network dynamics in collective tasks.    

%%%%%%%%%%%%%%%%%%%%%%%%%%%%%%%%%%%

\section*{Acknowledgement} The authors  thank Vishwesha Guttal, Colin Torney, Philip Holmes and Simon Levin for helpful discussion.  This work was supported in part by ARO grant W911NG-11-1-0385, ONR grant N00014-09-1-1074, NSF grant ECCS-1135724, and AFOSR grant FA9550-07-1-0-0528. 

%%%%%%%%%%%%%%%%%%%%%%%%%%%%%%%%%%% 

\appendix
\section{Adaptive dynamics calculations\label{s:app}}
In this section we discuss the details of the adaptive dynamics analysis from Section \ref{ss:ad}. Define the function $$ G(k_R,k_M) = \frac{k_M^2+(1-k_M)^2\beta^2(1-k_R)}{4(2k_M^2-2k_M+1)} +ck_M^2.$$
Then the differential fitness from \eqref{Sm} is given by $$S=\exp[-G(k_R,k_M)] - \exp[-G(k_R,k_R)].$$
The selection gradient $g(k_R)$ is given by 
\begin{align*}
g(k_R)=\left.\frac{\partial S}{\partial k_M}\right|_{k_M=k_R}&=-\exp[-G(k_R,k_R)] \times \\
&\left( \frac{k_M(1-k_M)[1+\beta^2(k_R-1)]}{2(2k_M^2-2k_M+1)^2} +2ck_M\right).
\end{align*}
Solving for the singular strategy condition $g(k_*)=0$ gives the expression \eqref{sing} $$k_*(1-k_*)[1+\beta^2(k_*-1)] + 4ck_*(2k_*^2-2k_*+1)^2 =0.$$
This expression has two sets of solutions that are plotted in Figure \ref{f:nice}. One set corresponds to $k_*=0$ and the other is defined implicitly by the equation 
\begin{equation}
c=\frac{(k_*-1)[1+\beta^2(k_*-1)]}{4(2k_*^2-2k_*+1)^2}=:\tilde{c}(k_*) .
\label{ctilde}
\end{equation}
To determine conditions for evolutionary branching we compute 
\begin{align*}
\left.\frac{\partial^2S}{\partial k_M^2} \right|_{k_M=k_R=k_*}&= -\exp \left[-ck_*^2 + \frac{\beta^2(k_*-1)^3-k_*^2}{4-8k_*+8k_*^2} \right]\times \\ 
&\frac{\left(1+\beta^2 (-1+k_*)\right) k_* \left(3-10 k_*+6 k_*^2\right)}{2 \left(1-2 k_*+2 k_*^2\right)^3}.
\end{align*} 
Hence the branching condition $\left.\frac{\partial^2S}{\partial k_M^2} \right|_{k_M=k_R=k_*}>0$ corresponds to 
\begin{equation}
\frac{\left(\beta^2 (1-k_*)-1\right) k_* \left(3-10 k_*+6 k_*^2\right)}{2 \left(1-2 k_*+2 k_*^2\right)^3}>0 .
\label{ineqcond}
\end{equation}
The zeros of the function above in the range $k_*\in[0,1]$ are $0$, $\frac{5-\sqrt{7}}{6}$, and $1-\frac{1}{\beta^2}$. A derivative test shows that the condition \eqref{ineqcond} is satisfied for $k_*\in\left(0, \frac{5-\sqrt{7}}{6}\right)$.

The critical cost parameter $c_1$ in Figure \ref{f:nice}, corresponds to the maximum singular value $k_*$ for branching and is given from \eqref{ctilde} by $c_1=\tilde{c}\left( \frac{5-\sqrt{7}}{6}\right)$. The parameter $c_2$ is determined by calculating the local maximum of the function $\tilde{c}(k_*)$  as seen in Figure \ref{f:nice}. We use the notation $\tilde{c}(k_{crit})=c_2$. $k_{crit}$ can be calculated analytically and is given by a particular root of a cubic equation. The analytical expression for $k_{crit}$ (and correspondingly $c_2$) is cumbersome and hence is left out of this text. Nonetheless, the sketch in Figure \ref{f:nice} clearly conveys the main ideas. 

Finally we discuss the convergence stability of the singular strategies. The convergence stability condition is given by \eqref{css}. For the $k_*=0$ singular strategy, $$ \left.\frac{\partial g}{\partial k_R}\right|_{k_R=0}=\exp\left[\frac{-\beta^2}{4}\right]\frac{\beta^2-1-4c}{2}, \text{ hence } \left.\frac{\partial g}{\partial k_R}\right|_{k_R=0}<0 \iff c>\frac{\beta^2-1}{4}.$$

For the second singular strategy curve defined implicitly by \eqref{ctilde}, the derivative term in the convergence stability condition evaluates to
$$\left.\frac{\partial g}{\partial k_R}\right|_{k_R=k_*} = \frac{k_* \left[-3+10 k_*-6 k_*^2-2 \beta^2 \left(2 k_*^3 -6 k_*^2+5 k_*-1\right)\right] e^{-G(k_*,k_*)}}{2 \left(1-2 k_*+2 k_*^2\right)^3}.$$
The interior root of $\left.\frac{\partial g}{\partial k_R}\right|_{k_R=k_*}=0$ is precisely the value $k_{crit}$ that maximizes $\tilde{c}(k_*)$ (since $c=\tilde{c}(k_*) \equiv g(k_*)=0$); this root corresponds to $c_2$ (see above). Hence one can verify that the singular strategies corresponding to the curve \eqref{ctilde} are stable for $k_*>k_{crit}$ and unstable for $k_*<k_{crit}$ as shown in Figure \ref{f:nice}. 

%%%%%%%%%%%%%%%%%%%%%%%%%%%%%%%%%%% 

%% The Appendices part is started with the command \appendix;
%% appendix sections are then done as normal sections
%% \appendix

%% \section{}
%% \label{}

%% References
%%
%% Following citation commands can be used in the body text:
%% Usage of \cite is as follows:
%%   \cite{key}          ==>>  [#]
%%   \cite[chap. 2]{key} ==>>  [#, chap. 2]
%%   \citet{key}         ==>>  Author [#]

%% References with bibTeX database:

\bibliographystyle{model1-num-names}
\bibliography{leader_bib}

%% Authors are advised to submit their bibtex database files. They are
%% requested to list a bibtex style file in the manuscript if they do
%% not want to use model1-num-names.bst.

\end{document}